\documentclass[aps,pre,showpacs,twocolumn,amsmath,superscriptaddress]{revtex4-1}
\usepackage{amsmath}
\usepackage{bm}
\usepackage{subfigure}
\usepackage{latexsym}
\usepackage{amssymb} %Óд˺ê°ü£¬Ôò¿ÉʹÓÃ\leqslant%
\usepackage{amsthm} %ÕƹÜÖ¤Ã÷»·¾³£¬ÓÐËüÔò¿ÉÒÔʹÓÃ\begin{proof}\end{proof}%
\usepackage{times} %ÕƹÜÎÄ×Ö×ÖÌ壬ÓÐËüÔòΪÐÂÂÞÂí×ÖÌå%
\usepackage{graphicx}

\newcommand{\aop}{Ann. Phys.~}

\newcommand{\jmp}{J. Math. Phys.~}
\newcommand{\jpa}{J. Phys. A~}

\newcommand{\njp}{New. J. Phys.~}

\newcommand{\pla}{Phys. Lett. A~}

\newcommand{\tinyspace}{\mspace{1mu}}

\newcommand{\abs}[1]{\left\lvert\tinyspace #1 \tinyspace\right\rvert}

\def\dif{\mathrm{d}}

\def\vol{\mathrm{vol}}

\def\complex{\mathbb{C}}
\def\real{\mathbb{R}}

\def\dif{\mathrm{d}}
\def\1{\mathbf{1}}

\def\ot{\otimes}

\newcommand{\inner}[2]{\langle #1 , #2\rangle}

\newcommand{\defeq}{\stackrel{\smash{\textnormal{\tiny def}}}{=}}

\newcommand{\Pa}[1]{\left(#1\right)}

\newcommand{\Br}[1]{\left[#1\right]}
\newcommand{\set}[1]{\{#1\}}
\newcommand{\Set}[1]{\left\{#1\right\}}

\def\cH{\mathcal{H}}\def\cJ{\mathcal{J}}

\def\bsa{\boldsymbol{a}}\def\bsb{\boldsymbol{b}}\def\bse{\boldsymbol{e}}

\def\bsp{\boldsymbol{p}}\def\bsq{\boldsymbol{q}}
\def\bsu{\boldsymbol{u}}\def\bsv{\boldsymbol{v}}\def\bsx{\boldsymbol{x}}\def\bsy{\boldsymbol{y}}

\def\rJ{\mathrm{J}}
\def\rM{\mathrm{M}}\def\rN{\mathrm{N}}

\def\A{\textsf{A}}\def\B{\textsf{B}}
\def\G{\textsf{G}}                  
\def\M{\textsf{M}}
\def\T{\textsf{T}}

\newtheorem{thrm}{Theorem}[section]

\newtheorem{prop}[thrm]{Proposition}

\theoremstyle{definition}

\begin{document}

%====================================================================================%

\title{Incompatibility probability of random quantum measurements}

%====================================================================================%

\author{Lin Zhang}
\email{godyalin@163.com} \affiliation{Institute of Mathematics,
Hangzhou Dianzi University, Hangzhou 310018, PR~China}
\affiliation{Max-Planck-Institute for Mathematics in the Sciences,
Leipzig 04103, Germany}
\author{Hua Xiang}
\email{hxiang@whu.edu.cn} \affiliation{School of Mathematics and
Statistics, Wuhan University, Wuhan 430072, PR China}
\author{Xianqing Li-Jost}
\email{xli-jost@mis.mpg.de} \affiliation{Max-Planck-Institute for
Mathematics in the Sciences, Leipzig 04103, Germany}
\author{Shao-Ming Fei}
\email{feishm@cnu.edu.cn} \affiliation{Max-Planck-Institute for
Mathematics in the Sciences, Leipzig 04103, Germany}
\affiliation{School of Mathematical Sciences, Capital Normal
University, Beijing 100048, China}

%\date{\today}

\begin{abstract}
Incompatibility of quantum measurements is of fundamental importance
in quantum mechanics. It is closely related to many nonclassical
phenomena such as Bell nonlocality, quantum uncertainty relations,
and quantum steering. We study the necessary and sufficient
conditions of quantum compatibility for a given collection of $n$
measurements in $d$-dimensional space. From the compatibility
criterion for two-qubit measurements, we compute the incompatibility
probability of a pair of independent random measurements. For a pair
of unbiased random qubit measurements, we derive that the
incompatibility probability is exactly $\frac35$. Detailed results
are also presented in figures for pairs of general qubit
measurements.
\end{abstract}

\maketitle

%=====================================================%
\section{Introduction}
%=====================================================%

Quantum theory has become the pillar of modern physics. Features
such as non-locality \cite{Bell1964}, steering \cite{Wiseman2007},
entanglement \cite{EPR1935,Werner1989}, contextuality
\cite{Bell1966}, uncertainty \cite{Heisenberg1927} and coherence
\cite{Baumgratz2014} distinguish quantum physics from classical
physics. Among these features of quantum physics, the quantum
incompatibility of quantum measurements forbids one from measuring
two observables simultaneously exactly when they are incompatible.
Quantum incompatibility can lead to many novel phenomena including
measurement uncertainty relations \cite{Busch2014}, steerability
\cite{Uola2015} and nonlocality \cite{Fine1982}. In the case of a
pair of two-outcome measurements, the incompatibility is equivalent
to Bell non-locality \cite{Wolf2009,Bene2018}, though measurement
incompatibility does not imply Bell non-locality in general
\cite{Quintino2016,Hirsch2018}.

Concerning quantum incompatibility, an important problem is the
development of an effective method to judge whether a set of
measurements is compatible (i.e., jointly measurable), which has
received much attention \cite{Bluhm2018,Bluhm2019,Jae2019}. The
authors of \cite{Bluhm2018} and \cite{Bluhm2019} used the notion of
free spectrahedra in the optimization theory to characterize the
measurement compatibility (also known as the joint measurability).
Due to the abstract construction of free spectahedra,
characterization of incompatibility along this approach is not very
operational. The authors of a recent study \cite{Jae2019} presented
a more operational way toward the characterization of quantum
incompatibility for the case where both measurements have the same
number of measurement outcomes.

The relations between quantum measurement incompatibility and
quantum information processing have also been extensively
investigated. In fact, it is shown that quantum incompatibility can
be detected by a state discrimination task with partial intermediate
information \cite{Carmelli2019,Carmelli2018}. Looking at it from
another perspective, every set of incompatible measurements provides
an advantage over compatible ones in a suitably chosen quantum state
discrimination task \cite{Uola2019,Skrzypczyk2019}.

In \cite{Zyczkowski1998} the separability probability problem has
been addressed: What is the probability that a randomly given
quantum state is entangled (or separable)? In order to answer this
question, the authors proposed to calculate the volume of all
separable bipartite states in a portion of the whole set of
bipartite states \cite{Zyczkowski2003} and \cite{Andai2006}. The
issue of probing entanglement and constructing a separable form in
the context of separable states has also been addressed in
\cite{Avron2009,Samuel2018}. Nevertheless, even for the simplest
case (i.e. two-qubit quantum states), computing the separability
probability according to the Hilbert-Schmidt measure is still a
challenging problem. Numerical simulations lead to intriguing
formulas for separability probability \cite{Slater2013}. It turned
out that the geometric separability probability of two-qubit quantum
systems is conjectured to be $\frac8{33}$, without proof up to now
\cite{Andai2017}.

Motivated by the problem of separability probability, we ask what
the probability is for a randomly given pair of measurements
[positive operator valued measurements (POVMs)] to be incompatible.
However, calculation of the incompatibility probability depends
heavily on the criteria of incompatibility. For a pair of unbiased
random qubit measurements with two measurement outcomes, we derive
that the incompatibility probability is exactly $\frac35$. The
incompatibility probability for a pair of general qubit measurements
is conjectured to be $\frac14$ by numerical simulation. As for the
case of a pair of two-outcome measurements, the incompatibility is
equivalent to Bell non-locality \cite{Bene2018}, this fact suggests
that 25\% of pairs of qubit measurements can lead to Bell
non-locality. If we are restricted to the use of pairs of unbiased
qubit measurements, the fraction increases to 60\%.

In this paper, first we deal with the necessary and sufficient
conditions of (in-)compatibility for a finite number of measurements
with arbitrary finite outcomes. Then we investigate the geometry of
the set of incompatible pairs of measurements. We compute the
incompatibility probability: the ratio of the set of incompatible
pairs of measurements versus the set of all pairs of measurements.

%===========================================================================%
\section{Characterization of quantum measurement incompatibility}\label{sect:quntincom}
%===========================================================================%

For a positive integer $\ell$, denote $[\ell]:=\set{1,\ldots,\ell}$.
We say that $\M=(M_{i_1\ldots i_n})$ is an $n$-th Hermitian tensor
if each $M_{i_1\ldots i_n}$ is a Hermitian operator acting on
$d$-dimensional Hilbert space $\cH_d$,
where $i_1\in[k_1],\ldots,i_n\in[k_n]$. A POVM is
represented by an $n$-th Hermitian tensor $\M=(M_{i_1\ldots i_n})$,
where each $M_{i_1\ldots i_n}$ is positive semi-definite and
$\sum M_{i_1\ldots i_n}=\1_d$, with $\1_d$ the identity operator on $\cH_d$.
The following $n$ POVMs $\A^{(1)}=(A^{(1)}_{i_1})$, $\ldots$, $\A^{(n)}=(A^{(n)}_{i_n})$, with $i_1\in[k_1]$, $\ldots$, $i_n\in[k_n]$, defined by
\begin{eqnarray*}
A^{(1)}_{i_1} \defeq \sum_{i_2\ldots i_n}M_{i_1\ldots
i_n},~\ldots,~A^{(n)}_{i_n} \defeq \sum_{i_1\ldots
i_{n-1}}M_{i_1\ldots i_n},
\end{eqnarray*}
are called the \emph{marginals} of $\M$.

For $n$ given POVMs $\A^{(l)}$ $(l\in[n])$ on $\cH_d$, if there exists a POVM,
$\M=(M_{i_1\ldots i_n})$, where $i_l\in[k_l]$  ($l \in [n]$), such that $\A^{(l)}$
$(l\in[n])$ are the marginals of $\M$, we say that the $n$ POVMs
$\A^{(l)}$ $(l\in[n])$ are \emph{compatible} or \emph{jointly
measurable}, and $\M$ is called the \emph{joint measurement}.
Otherwise, they are called \emph{incompatible} \cite{Bluhm2019}.

Generally, for given $n$ jointly measurable POVMs $\A^{(l)}$
$(l\in[n])$, their joint measurements are not unique. We denote
$\cJ\Pa{\A^{(1)},\ldots,\A^{(n)}}$ all the joint measurements for
$n$ arbitrary given POVMs $\A^{(l)}$ $(l\in[n])$. Then
$\cJ\Pa{\A^{(1)},\ldots,\A^{(n)}}=\emptyset$ if the $n$ POVMs
$\A^{(l)}$ $(l\in[n])$ are not jointly measurable. Thus $\A^{(l)}$
$(l\in[n])$ are jointly measurable if and only if
$\cJ\Pa{\A^{(1)},\ldots,\A^{(n)}}\neq\emptyset$.

Consider $n$ arbitrary probability vectors $\bsp^{(1)}
=\Pa{p^{(1)}_{i_1}},\ldots,\bsp^{(n)}=\Pa{p^{(n)}_{i_n}}$, where
$i_l\in[k_l]$  ($l \in [n]$), such that
%$i_1\in[k_1],\ldots,i_n\in[k_n]$.
all the components of $p^{(l)}$ are positive for all $l\in[n]$.
Denote $\bsp^{(l)}\ot\1_d$ a POVM with measurement operators given
by $\{p^{(l)}_{i_l}\1_d\}$, where $i_l\in[k_l]$, $l \in [n]$.
Clearly,
$\cJ\Pa{\bsp^{(1)}\ot\1_d,\ldots,\bsp^{(n)}\ot\1_d}\neq\emptyset$.
Let $\T=(T_{i_1\ldots i_n})$ be the $n$-th Hermitian tensor such
that their $n$ marginals are given by
$\bsp^{(1)}\ot\1_d,\ldots,\bsp^{(n)}\ot\1_d$,
\begin{eqnarray*}\label{eq:T-tensor}
\sum_{i_2,\ldots,i_k}T_{i_1\ldots i_n} =
p^{(1)}_{i_1}\1_d,~\ldots,~\sum_{i_1,\ldots,i_{n-1}}T_{i_1\ldots i_n}
= p^{(n)}_{i_n}\1_d.
\end{eqnarray*}
Namely, $\T\in\cJ\Pa{\bsp^{(1)}\ot\1_d,\ldots,\bsp^{(n)}\ot\1_d}$. We have

{\bf Theorem 1.} For $n$ POVMs $\A^{(l)}$ and $n$ probability
vectors $\bsp^{(l)}$, set
\begin{eqnarray}\label{thm1}
M_{i_1\ldots i_n} =
\Pa{\prod^n_{\ell=1}p^{(\ell)}_{i_\ell}}\sum^n_{l=1}\frac1{p^{(l)}_{i_l}}A^{(l)}_{i_l}
- (n-1)T_{i_1\ldots i_n}.
\end{eqnarray}
Then $\A^{(l)}$ $(l\in[n])$ are $n$ marginals of $\M=(M_{i_1\ldots
i_n})$, and $\A^{(l)}$ $(l\in[n])$ are compatible if and only if for
any collection of $n$ probability vectors $\bsp^{(l)}(l\in[n])$,
there exists some $n$-th Hermitian tensor $\T\in
\cJ\Pa{\bsp^{(1)}\ot\1_d,\ldots,\bsp^{(n)}\ot\1_d}$ such that $\M$
is a POVM.

\begin{proof}[{\bf Proof.}]
Define $\G=(G_{i_1\ldots i_n})$ as follows:
\begin{eqnarray*}
G_{i_1\ldots i_n}\defeq
\frac{\prod^n_{\ell=1}p^{(\ell)}_{i_\ell}}n\sum^n_{l=1}\frac1{p^{(l)}_{i_l}}A^{(l)}_{i_l}.
\end{eqnarray*}
It is directly verified that $\G$ is a POVM. Its $n$ marginals are given by
\begin{eqnarray*}
\frac1n\A^{(l)}+\Pa{1-\frac1n}\bsp^{(l)}\ot\1_d\quad (l\in[n]).
\end{eqnarray*}
Here we view each POVM as a column-block matrix and $\ot$ stands for
the Kronecker tensor product. That is,
$$
\A^{(l)}=\Pa{\begin{array}{c}
                          A^{(l)}_1 \\
                          \vdots \\
                          A^{(l)}_{k_l}
                        \end{array}
}\text{ and }~\bsp^{(l)}\ot\1_d = \Pa{\begin{array}{c}
                          p^{(l)}_1\1_d \\
                          \vdots \\
                          p^{(l)}_{k_l}\1_d
                        \end{array}
}.
$$
%----------------------------------------------------------------%
Apparently $G_{i_1\ldots i_n}$ is non-negative for all $i_l\in[k_l]$
($l \in [n]$) by definition. Moreover, one has
\begin{eqnarray*}
\sum_{i_2,\ldots,i_n}G_{i_1\ldots
i_n}=\frac1nA^{(1)}_{i_1}+\frac1np^{(1)}_{i_1}\1_d+\cdots+\frac1np^{(1)}_{i_1}\1_d,
\end{eqnarray*}
and hence $\sum_{i_1,\ldots,i_n}G_{i_1\ldots i_n} =\1_d$. This shows that
$\frac1n\A^{(1)}+\Pa{1-\frac1n}\bsp^{(1)}\ot\1_d$ is one of the
marginals of $\G$. Other marginals can be obtained similarly.
Furthermore, the $n$-th Hermitian tensor $n\G -
(n-1)\T=(nG_{i_1\ldots i_n} - (n-1)T_{i_1\ldots i_n})$ has $n$
marginals $\A^{(1)},\ldots,\A^{(n)}$. Indeed, for instance,
$\sum_{i_2,\ldots,i_n}(nG_{i_1\ldots i_n} - (n-1)T_{i_1\ldots i_n})
=A^{(1)}_{i_1}$. This completes the proof.
\end{proof}

Theorem 1 also indicates that adding noise to the POVMs, i.e.,
taking convex combinations of the original measurement operators of these POVMs and the trivial
measurement (the identity operator), can make the resulting new POVMs more compatible (jointly measurable).

In the following we consider the case of $d=n=2$.
Let $\A=(A_1,A_2)$ and $\B=(B_1,B_2)$ be two POVMs on $\complex^2$.
By using the Bloch representation, we can generally write
\begin{eqnarray}\label{eq:APOVM}
A_i=\frac12\Br{(1+(-1)^ia_0)\1_2+(-1)^{i}\bsa\cdot\boldsymbol{\sigma}},~~i=1,2,
\end{eqnarray}
where $\boldsymbol{\sigma}=(\sigma_1,\sigma_2,\sigma_3)$ with
$\sigma_i$, i=1,2,3, the Pauli matrices, and $\bsa$ is a three
dimensional real vector satisfying $\abs{\bsa}\leqslant 1-\abs{a_0}$
with $a_0\in[-1,1]$. Here $\abs{\bsa}$ is referred to as the
\emph{sharpness} while $\abs{a_0}$ the \emph{biasedness}. Similarly,
\begin{eqnarray}\label{eq:BPOVM}
B_j =
\frac12\Br{(1+(-1)^jb_0)\1_2+(-1)^j\bsb\cdot\boldsymbol{\sigma}},~~
j=1,2,
\end{eqnarray}
where $\abs{\bsb}\leqslant 1-\abs{b_0}$ with $b_0\in[-1,1]$.
Choosing arbitrarily two probability vectors $\bsp=(p_1,p_2)$ and
$\bsq=(q_1,q_2)$, we have from Eq.~\eqref{thm1},
\begin{eqnarray*}
\M(\A,\B;\bsp,\bsq;\T) = (M_{ij}),
\end{eqnarray*}
where $M_{ij}=q_jA_i+p_iB_j-T_{ij}$, and $\T=(T_{ij})$ satisfies
that $T_{i1}+T_{i2}=p_i\1_2$ and $T_{1j}+T_{2j}=q_j\1_2$ for
$i,j\in[2]$.

We write $\T$ in the block-matrix form:
\begin{eqnarray*}
\T=\Pa{\begin{array}{cc}
                                          X & p_1\1_2-X \\
                                          q_1\1_2-X &
                                          X+(q_2-p_1)\1_2
                                        \end{array}
},
\end{eqnarray*}
where $p_2-q_1=q_2-p_1$, and $X$ is some $2\times 2$ Hermitian
matrix. Assume that $p_1=p$ and $q_1=q$, where $p,q\in[0,1]$. By
Bloch representation, $X$ can be written as
\begin{eqnarray*}
X = \frac12\Br{(1-x_0)\1_2-\bsx\cdot\boldsymbol{\sigma}},\quad
(x_0,\bsx)\in\real^4.
\end{eqnarray*}
Note that $\M$ is a legal POVM if and only if $M_{ij}\geqslant0$ for
all $i,j$. In other words, $\M$ is a POVM if and only if
\begin{eqnarray*}
\begin{cases}
&\abs{q\bsa+p\bsb-\bsx}\\
& \quad \leqslant (q+p-1)-(qa_0+pb_0-x_0),\\
&\abs{\bsx+(1-q)\bsa-p\bsb}\\
&\quad \leqslant (2-q-p)-(x_0+(1-q)a_0-pb_0),\\
&\abs{\bsx-q\bsa+(1-p)\bsb}\\
& \quad \leqslant (2-q-p)-(x_0-qa_0+(1-p)b_0),\\
& \abs{(1-q)\bsa+(1-p)\bsb+\bsx}\notag\\
& \quad \leqslant (q+p-1)+(x_0+(1-q)a_0+(1-p)b_0).
\end{cases}
\end{eqnarray*}

A question naturally arises is: What kind of relationship should be
satisfied by the 8-tuple $(a_0,\bsa,b_0,\bsb)$ such that $\A$ and
$\B$ are jointly measurable for any prescribed $\bsp$ and $\bsq$.
Since $\T$ is $\bsp$ and $\bsq$ dependent, without loss of
generality, we may take probability vectors
$\bsp=\bsq=(\frac12,\frac12)$. Denote $\bsu=\bsa+\bsb$,
$\bsv=\bsa-\bsb$; $\alpha=a_0+b_0$, $\beta=a_0-b_0$, and
$\bsy=2\bsx$, $y_0=2x_0$. Now for a pair of unbiased observables
$\A$ and $\B$, i.e., $a_0=b_0=0$, one has $\alpha=\beta=0$ and the
solution set of the above four inequalities is not empty if and only
if $y_0\geqslant \abs{\bsu}\text{ and }2-y_0\geqslant \abs{\bsv}$,
i.e., $y_0$ in the closed interval $\Br{\abs{\bsu},2-\abs{\bsv}}$.
This amounts to saying that $\abs{\bsu}\leqslant2-\abs{\bsv}$.
Therefore $\A$ and $\B$ are jointly measurable if and only if
$\abs{\bsa+\bsb}+\abs{\bsa-\bsb}\leqslant2$, which is just the
result obtained in \cite{Busch1986}. For a general pair of qubit
observables, the following result \cite{Yu2010} answers this
question. For a pair of qubit observables $\A$ and $\B$ in
Eq.~\eqref{eq:APOVM} and Eq.~\eqref{eq:BPOVM}, respectively, $\A$
and $\B$ are compatible if and only if the following inequality
holds:
\begin{eqnarray}\label{eq:yu}
&&\Pa{1-h(a_0,\bsa)^2-h(b_0,\bsb)^2}\Pa{1-\frac{a^2_0}{h(a_0,\bsa)^2}-\frac{b^2_0}{h(b_0,\bsb)^2}}\notag\\
&&\leqslant (\inner{\bsa}{\bsb}-a_0b_0)^2,
\end{eqnarray}
where
$h(x_0,\bsx)=\frac{\sqrt{(1+x_0)-\abs{\bsx}^2}+\sqrt{(1-x_0)-\abs{\bsx}^2}}2$
for $x_0\in[-1,1]$ and $\abs{\bsx}\leqslant1-\abs{x_0}$.

Based on this result, in what follows, we analyze the
incompatibility probability of random qubit measurements.

%=================================================================================================%
\section{Incompatibility probability of random qubit measurements}\label{sect:incomprob}
%=================================================================================================%

We now consider the following question. Let $\Theta_{d,n}$ be the
set of all pairs $(\A,\B)$ of POVMs with $n$ measurement operators
each, with $\A=(A_i)^n_{i=1}$ and $\B=(B_j)^n_{j=1}$ acting on
$\complex^d$. Denote $\Theta^{\rN\rJ\rM}_{d,n}$ and
$\Theta^{\rJ\rM}_{d,n}$ the set of all incompatible and compatible
pairs of POVMs from $\Theta_{d,n}$, respectively. Namely,
$\Theta^{\rJ\rM}_{d,n}=\Theta_{d,n}\backslash\Theta^{\rN\rJ\rM}_{d,n}$.
Let $\vol(\Theta^{\rN\rJ\rM}_{d,n})$ and $\vol(\Theta_{d,n})$ be the
volumes of $\Theta^{\rN\rJ\rM}_{d,n}$ and $\Theta_{d,n}$,
respectively. We would like to know the geometric probability of
incompatibility,
$\mathbf{Pr}[\Theta^{\rN\rJ\rM}_{d,n}]=\vol(\Theta^{\rN\rJ\rM}_{d,n})/\vol(\Theta_{d,n})$.
This question heavily depends on the criteria of compatibility. We
treat this problem below for the case of qubit POVMs.

In the following we study the geometric probability of
incompatibility for fixed $(a_0,b_0)$. For fixed
$(a_0,b_0)\in[-1,1]\times[-1,1]$ in POVMs Eq.~\eqref{eq:APOVM} and
Eq.~\eqref{eq:BPOVM}, we denote $\Theta_{2,2}(a_0,b_0)$ the section
of $\Theta_{2,2}$ at $(a_0,b_0)$, and similarly for
$\Theta^{\rN\rJ\rM}_{2,2}(a_0,b_0)$ and
$\Theta^{\rJ\rM}_{2,2}(a_0,b_0)$. We consider the parameterized
probabilities:
$\mathbf{Pr}[\Theta^{\rN\rJ\rM}_{2,2}(a_0,b_0)]=\vol(\Theta^{\rN\rJ\rM}_{2,2}(a_0,b_0))/\vol(\Theta_{2,2}(a_0,b_0))$
and
$\mathbf{Pr}[\Theta^{\rJ\rM}_{2,2}(a_0,b_0)]=1-\mathbf{Pr}[\Theta^{\rN\rJ\rM}_{2,2}(a_0,b_0)]$.
Note that the parameters $\bsa$ and $\bsb$ in
$\Theta_{2,2}(a_0,b_0)$ satisfy the constraints $\abs{\bsa}\leqslant
1-\abs{a_0}$ and $\abs{\bsb}\leqslant 1-\abs{b_0}$. We have
\begin{eqnarray*}
\vol(\Theta_{2,2}(a_0,b_0))=\Pa{\frac{4\pi}3}^2(1-\abs{a_0})^3(1-\abs{b_0})^3.
\end{eqnarray*}
It suffices to calculate the volume $\vol(\Theta^{\rN\rJ\rM}_{2,2}(a_0,b_0))$.

\subsection{The case for unbiased measurements: $(a_0,b_0)=(0,0)$}

We first consider the unbiased POVMs $\A$ and $\B$, i.e.,
$(a_0,b_0)=(0,0)$, determined by the vectors $\bsa$ and $\bsb$,
respectively. In this case Eq.~\eqref{eq:yu} gives rise to the
condition that $\A$ and $\B$ are incompatible:
$f(\bsa,\bsb):=\abs{\bsa}^2+\abs{\bsb}^2-s^2>1$, where
$s=\inner{\bsu}{\bsv}\in[-1,1]$. This condition $f(\bsa,\bsb)>1$ is
equivalent to $g(\bsa,\bsb):=\abs{\bsa+\bsb}+\abs{\bsa-\bsb}>2$ (see
Appendix~\ref{sect:app-A}).

In stead of calculating the volume
$\vol(\Theta^{\rN\rJ\rM}_{2,2}(0,0))$, here we can also consider
$\bsa$ and $\bsb$ as random vectors with probability distribution
$\dif \omega=p(\bsa)p(\bsb)[\dif \bsa][\dif \bsb]$ given by \cite{Zhang2018pla},
\begin{eqnarray}\label{eq:pdf}
\dif\omega=\Pa{\frac3{4\pi}}^2a^2b^2\delta(1-\abs{\bsu})\delta(1-\abs{\bsv})\dif
a\dif b [\dif \bsu] [\dif \bsv],
\end{eqnarray}
where $\bsa=a\bsu$ and $\bsb=b\bsv$ with $a=\abs{\bsa}\in[0,1]$, $b=\abs{\bsb}\in[0,1]$,
and $\abs{\bsu}=\abs{\bsv}=1$.

Denote the whole domain corresponding to $\Theta_{2,2}(0,0)$ by
$\widetilde\Omega=\Set{(\bsa,\bsb)\in\real^3\times
\real^3:\abs{\bsa}\leqslant1\text{ and }\abs{\bsb}\leqslant1}$ and
the domain corresponding to $\Theta^{\rN\rJ\rM}_{2,2}(0,0)$ by the
following
\begin{eqnarray*}
\widetilde\Omega_{\rN\rJ\rM}&=&\Set{(\bsa,\bsb)\in\real^6:f(\bsa,\bsb)>1\wedge
\abs{\bsa}\leqslant1\wedge
\abs{\bsb}\leqslant1}\\
&=&\Set{(\bsa,\bsb)\in\real^6:g(\bsa,\bsb)>2\wedge
\abs{\bsa}\leqslant1\wedge \abs{\bsb}\leqslant1}.
\end{eqnarray*}
It is easily verified that $\int_{\widetilde\Omega} \dif\omega=1$.
The problem is to calculate $\int_{\widetilde\Omega_{\rN\rJ\rM}}
\dif\omega$.

The condition $f(\bsa,\bsb)>1$ or $g(\bsa,\bsb)>2$ can be expressed
as $s^2<a^{-2}+b^{-2}-(ab)^{-2}$, which can be rewritten as
\begin{eqnarray}\label{pp}
\begin{cases}
a\in(0,1), b\in\Pa{\sqrt{1-a^2},1},\\
s\in\Pa{-\frac{\sqrt{a^2+b^2-1}}{ab},\frac{\sqrt{a^2+b^2-1}}{ab}}.
\end{cases}
\end{eqnarray}
The joint probability density function of such a 3-tuple
$(a,b,s)\in[0,1]^2\times[-1,1]$ is given by
$$
p(a,b,s)=3a^2\times 3b^2\times \frac12=\frac92 a^2b^2.
$$
Denote
$\Omega_{\rN\rJ\rM}=\Set{(a,b,s)\in\Omega:s^2<a^{-2}+b^{-2}-(ab)^{-2}}$,
i.e., all 3-tuples $(a,b,s)$ corresponding to $f(\bsa,\bsb)>1$ [or
$g(\bsa,\bsb)>2$]. Now the 3-tuple $(a,b,s)\in\Omega_{\rN\rJ\rM}$ if
and only if the conditions in Eq.~\eqref{pp} are satisfied. In fact,
$(a,b,s)\in\Omega_{\rN\rJ\rM}$ can also be rewritten as
\begin{eqnarray*}
s\in(-1,1), ~~ a\in(0,1), ~~ b\in\Pa{\sqrt{\frac{1-a^2}{1-
a^2s^2}},1}.
\end{eqnarray*}
Now the problem of calculating $\int_{\widetilde\Omega_{\rN\rJ\rM}}
\dif\omega$ is reduced to the calculation of
$\int_{\Omega_{\rN\rJ\rM}} p(a,b,s)\dif a\dif b\dif s$. We have
\begin{eqnarray*}
\int_{\Omega_{\rN\rJ\rM}} p(a,b,s)\dif a\dif b\dif s=\frac35.
\end{eqnarray*}
Therefore, we have

{\bf Theorem 2.} For a pair of random unbiased qubit POVMs $\A$ and
$\B$, generated by Eq.~\eqref{eq:pdf} via Bloch representation, the
incompatibility probability is given by
\begin{eqnarray}\label{eq:unbiasedpovm}
\mathbf{Pr}\Br{\Theta^{\rN\rJ\rM}_{2,2}(0,0)}=\frac35.
\end{eqnarray}
{\bf Remark.} Equation~\eqref{eq:unbiasedpovm} can also be derived
by calculating the volume $\vol(\Theta^{\rN\rJ\rM}_{2,2}(0,0))$.
Using the Lebesgue measure, we have $[\dif \bsa]=a^2\dif a\times
\delta(1-\abs{\bsu})[\dif\bsu]$. The Lebesgue volume of $\widetilde
\Omega$ is given by $\vol(\widetilde\Omega)=\frac{(4\pi)^2}9$. Thus
$\widetilde\Omega_{\rN\rJ\rM}$ can be expressed as
\begin{widetext}
$$
\Set{(a\bsu,b\bsv)\in\widetilde\Omega:a\in(0,1),~b\in
\Pa{\sqrt{1-a^2},1},~s\in\Pa{-\frac{\sqrt{a^2+b^2-1}}{ab},\frac{\sqrt{a^2+b^2-1}}{ab}}}.
$$
Then
\begin{eqnarray*}
\vol(\widetilde\Omega_{\rN\rJ\rM})&=&\int_{\widetilde\Omega_{\rN\rJ\rM}}[\dif\bsa][\dif\bsb]=N^2_3\int^1_0
\dif a\, a^2\int^1_{\sqrt{1-a^2}}\dif b\,
b^2\int^{\frac{\sqrt{a^2+b^2-1}}{ab}}_{-\frac{\sqrt{a^2+b^2-1}}{ab}}p_3(s)\dif
s \\[2mm]
&=&\frac{(4\pi)^2}2 \int^1_0 \dif a\, a^2\int^1_{\sqrt{1-a^2}}\dif b\,
b^2\int^{\frac{\sqrt{a^2+b^2-1}}{ab}}_{-\frac{\sqrt{a^2+b^2-1}}{ab}}\dif
s \\
&=& (4\pi)^2 \int^1_0 \dif a\, a^2\int^1_{\sqrt{1-a^2}}\dif b\,
b^2\frac{\sqrt{a^2+b^2-1}}{ab}\\
&=&(4\pi)^2 \int^1_0 \dif a\, a^2\frac{a^2}3\\
&=& \frac{(4\pi)^2}{15},
\end{eqnarray*}
where $N_3=4\pi$ and $p_3(s)=\frac12$. Note that $p_3(s)$ is just
the case of $p_m(s)$ for $m=3$, and $p_m(s)$ is given by (see
Appendix~\ref{sect:app-B})
$$
p_m(s)=\frac1{N^2_m}\int_{\real^m\times\real^m}\delta(s-\inner{\bsu}{\bsv})\delta(1-\abs{\bsu})
\delta(1-\abs{\bsv})[\dif\bsu][\dif\bsv].
$$
\end{widetext}
Therefore, the geometric probability of incompatibility is given by $\vol(\widetilde
\Omega_{\rN\rJ\rM})/\vol(\widetilde\Omega)=3/5$.

As $\A$ and $\B$ are incompatible if $f(\bsa,\bsb)>1$ [or
$g(\bsa,\bsb)>2$], it is also interesting to calculate analytically
the expectations $\mathbb{E}[f(\bsa,\bsb)]$ and
$\mathbb{E}[g(\bsa,\bsb)]$ of $f(\bsa,\bsb)$ and $g(\bsa,\bsb)$,
respectively. By direct computation we have
\begin{eqnarray*}
\mathbb{E}[f(\bsa,\bsb)]&=&\int_{\Omega_{\rN\rJ\rM}}
f(a,b,s)p(a,b,s)\dif a\dif b\dif s \\
&=& \frac{27}{25}>1
\end{eqnarray*}
and, similarly, $\mathbb{E}[g(\bsa,\bsb)]=\frac{72}{35}>2$. These results are
consistent with Eq.~\eqref{eq:unbiasedpovm}: two randomly selected
measurements $\A$ and $\B$ are most probably incompatible.

\subsection{The case $(a_0,b_0)=(\lambda,0)$ for $\lambda\in(-1,1)$}

The case $(a_0,b_0)=(\lambda,0)$, where $\lambda\in(-1,1)$,
corresponds to the case where $\A$ is a biased measurement and $\B$
is an unbiased one. In this case
$\Theta^{\rN\rJ\rM}_{2,2}(\lambda,0)$ can be parameterized as the
set $\widetilde\Theta^{\rN\rJ\rM}_{2,2}(\lambda,0)$ such that
\begin{widetext}
\begin{eqnarray*}
&&b\in(\sqrt{\abs{\lambda}},1), \quad
s\in\Pa{-\sqrt{\frac{b^2-\abs{\lambda}}{b^2(1-\abs{\lambda})}},
\sqrt{\frac{b^2-\abs{\lambda}}{b^2(1-\abs{\lambda})}}},\\
&&a\in\Pa{\sqrt{\frac{b^2 - b^4 - b^2 s^2 + b^4 s^2 - \abs{\lambda}^2
+ b^2 \abs{\lambda}^2 + b^2 s^2\abs{\lambda}^2 -
 b^4 s^2 \abs{\lambda}^2}{b^2(1-s^2)(1-b^2s^2)}},1-\abs{\lambda}}.
\end{eqnarray*}
\end{widetext}
Thus,
\begin{eqnarray*}
\mathbf{Pr}[\Theta^{\rN\rJ\rM}_{2,2}(\lambda,0)] =
\frac92(1-\abs{\lambda})^{-3}\int_{\widetilde
\Theta^{\rN\rJ\rM}_{2,2}(\lambda,0)} a^2b^2\dif b\dif s\dif a.
\end{eqnarray*}

By numerical computation it can be shown that
$\mathbf{Pr}[\Theta^{\rN\rJ\rM}_{2,2}(\lambda,0)]$ decreases when
$\abs{\lambda}\in[0,1)$ increases. Namely for larger
$\abs{\lambda}$, randomly selected $\A$ and $\B$ are most probably
compatible.

\subsection{The case of general $(a_0,b_0)$}

Generally, $\Theta_{2,2}$ can be identified as
$(a_0,a\bsu,b_0,b\bsv)$ such that $\abs{a_0}+a\leqslant1$ with
$a_0\in[-1,1],\,a\in[0,1]$ and $\abs{b_0}+b\leqslant1$ with
$b_0\in[-1,1],\,b\in[0,1]$. $\Theta^{\rN\rJ\rM}_{2,2}$ is a subset
of $\Theta_{2,2}$ and can be identified as the set that
Eq.~\eqref{eq:yu} is violated.

Denote $V=\Set{(x_0,\bsx)\in\real^4: \abs{x_0}+\abs{\bsx}\leqslant1\text{
for }x_0\in[-1,1]}$. We have
\begin{eqnarray*}
\vol(V) &=& \int^1_{-1}\dif
x_0\int_{\Set{\bsx\in\real^3:\abs{\bsx}\leqslant 1-\abs{x_0}}}[\dif
\bsx]\\
&=&\int^1_{-1}\frac{4\pi (1-\abs{x_0})^3}{3}\dif x_0=\frac{2\pi}3.
\end{eqnarray*}
It is easily seen that
$\vol(\Theta_{2,2})=\vol(V)^2=\Pa{\frac{2\pi}{3}}^2$. Thus the
incompatibility probability is given by
$$
\mathbf{Pr}[\Theta^{\rN\rJ\rM}_{2,2}]=\frac{\vol(\Theta^{\rN\rJ\rM}_{2,2})}{\vol(\Theta_{2,2})}.
$$
The volume of $\Theta^{\rN\rJ\rM}_{2,2}$ can be obtained as follows
by numerical calculation (see Appendix~\ref{sect:app-C}):
$\vol(\Theta^{\rN\rJ\rM}_{2,2}) \doteq 1.09662$, which is
approximately $\frac{\pi^2}9$. We \emph{conjecture} that the
incompatibility probability of a pair of random qubit measurements
$\A$ and $\B$, generated by $(a_0,\bsa)$ and $(b_0,\bsb)$ in
Eq.~\eqref{eq:yu}, is given by
\begin{eqnarray}
\mathbf{Pr}[\Theta^{\rN\rJ\rM}_{2,2}]=\frac14.
\end{eqnarray}

Moreover, the following result can be found:
if $\abs{a_0}=\abs{b_0}\geqslant\frac12$, then
$\Theta^{\rN\rJ\rM}_{2,2}(a_0,b_0)=\emptyset$, i.e., $\A$ and $\B$
are compatible, $\mathbf{Pr}[\Theta^{\rN\rJ\rM}_{2,2}(a_0,b_0)]=0$.

More detailed computational results on the incompatibility
probability $\mathbf{Pr}[\Theta^{\rN\rJ\rM}_{2,2}(a_0,b_0)]$ are
displayed in Figs.~\ref{Fig_njmprob}--\ref{Fig_onefourth}. From the
Fig.~\ref{Fig_njmprob} we observe that outside the curve
$\abs{a_0}+\abs{b_0} = 1$, the incompatibility probability is $0$.
It increases smoothly towards the origin where it attains the peak
value $\tfrac 35$. Figure~\ref{Fig_njmprobb} shows the contours of
Fig.~\ref{Fig_njmprob}, displaying isolines of the incompatibility
probability. Figure~\ref{Fig_onefourth} is one quarter of
Fig.~\ref{Fig_njmprob}, corresponding to the parameter regions
$a_0\geqslant0$ and $ b_0\geqslant0$.

\begin{figure}[htbp] \centering
\includegraphics[width=0.5\textwidth]{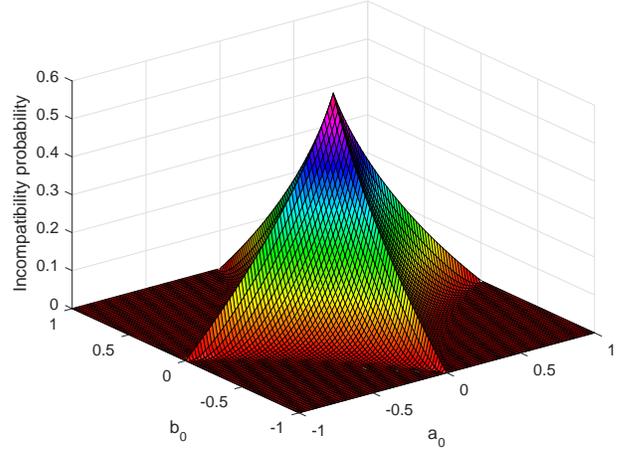}
\caption{The incompatibility probability
$\mathbf{Pr}[\Theta^{\rN\rJ\rM}_{2,2}(a_0,b_0)]$ of a pair of qubit
measurements $\A$ and $\B$, $(a_0,b_0)\in[-1,1]^2$.}
\label{Fig_njmprob}
\end{figure}

\begin{figure}[htbp] \centering
\includegraphics[width=0.5\textwidth]{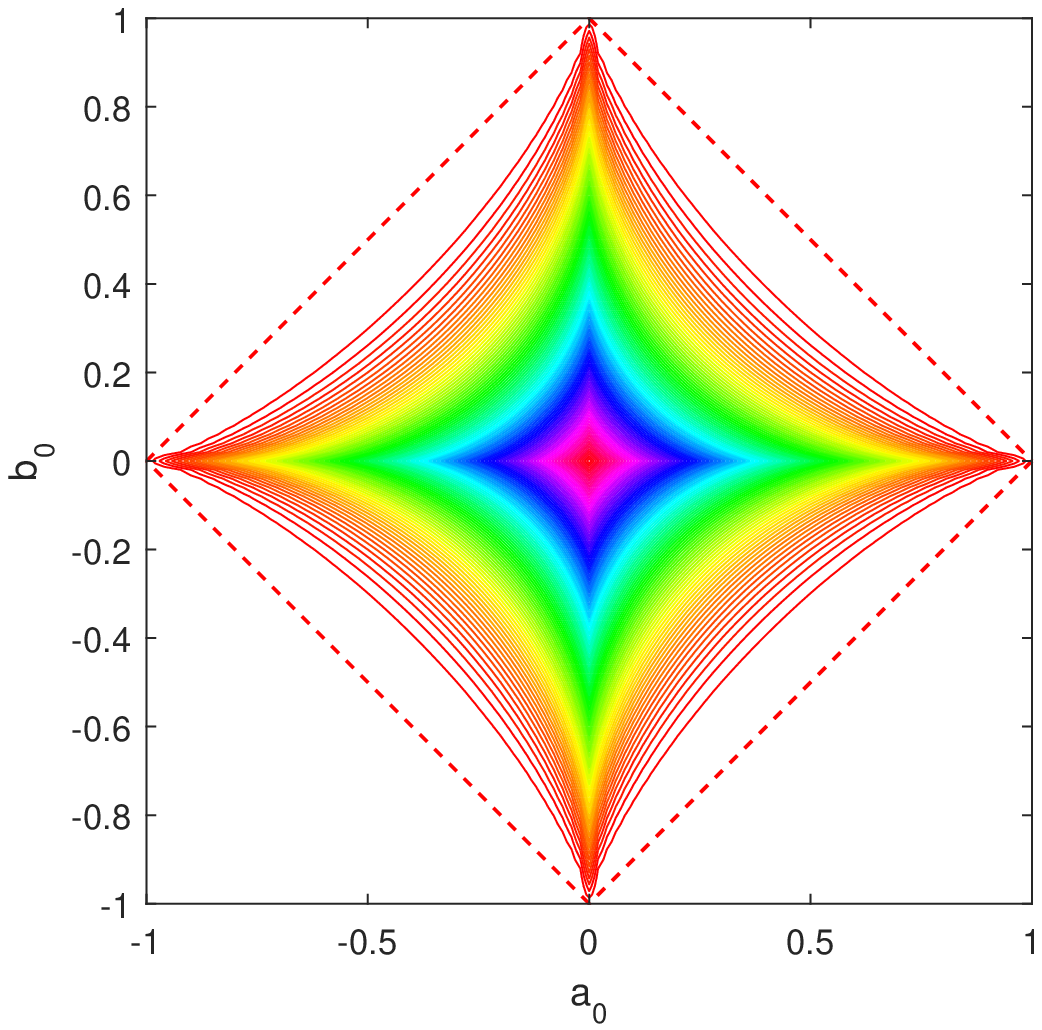}
\caption{The contours of the incompatibility probability
$\mathbf{Pr}[\Theta^{\rN\rJ\rM}_{2,2}(a_0,b_0)]$ corresponding to
Fig.~\ref{Fig_njmprob}.} \label{Fig_njmprobb}
\end{figure}

\begin{figure}[htbp] \centering
\includegraphics[width=0.5\textwidth]{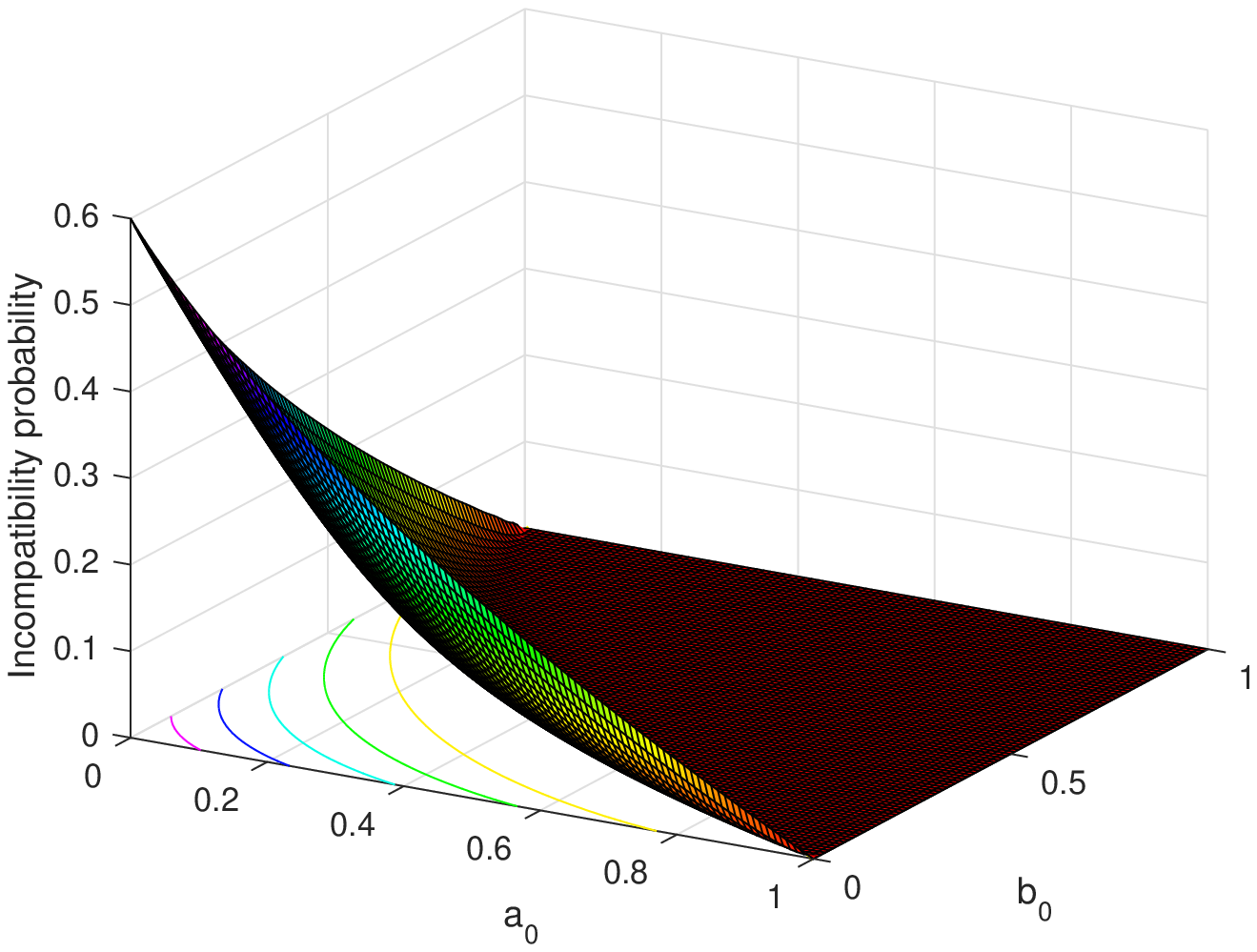}
\caption{The incompatibility probability
$\mathbf{Pr}[\Theta^{\rN\rJ\rM}_{2,2}(a_0,b_0)]$ of a pair of qubit
measurements $\A$ and $\B$ for $a_0\geqslant0$ and $
b_0\geqslant0$.} \label{Fig_onefourth}
\end{figure}

%======================================================================%
\section{Conclusions}\label{sect:conclusion}
%======================================================================%

We have dealt with the necessary and sufficient conditions of
incompatibility for a finite number of measurements with arbitrary
finite outcomes. Our approach toward quantum incompatibility covers
essentially the theoretic framework of \cite{Jae2019} and connects
with the results in \cite{Busch1986} and \cite{Yu2010} in the case
of qubit measurements with two-outcomes. Based on the necessary and
sufficient conditions of compatibility for qubit measurements, we
have analytically worked out the incompatibility probability of a
pair of unbiased qubit measurements. The incompatibility probability
of one unbiased and one biased qubit measurements, together with a
pair of general qubit measurements, has also been investigated by
analytical derivations and numerical calculations. These results may
highlight studies on topics such as Bell nonlocality, quantum
uncertainty and quantum steering. In fact, our results suggest that,
the possibility of a pair of unbiased qubit measurements leading to
Bell nonlocality is 60\%, which is larger than the 25\% for general
qubit measurements. Our results also indicate that the class of
qubit unbiased measurements is the best choice for studying the
connection between quantum measurement incompatibility and Bell
non-locality \cite{Quintino2016,Hirsch2018}. It would also be
interesting to design a schematic experiment to test the results.
For such experimental verification, one needs to construct random
gates to implement a random pair of POVMs $A$ and $B$ under a
specified distribution. Besides qubit measurements, it is also
interesting to estimate the incompatibility probability of a pair of
high-dimensional measurements by using the convex geometry and
probabilistic tools \cite{Aubrun2006,Aubrn2014}.

%=============================================================================%
\begin{acknowledgments}
LZ would like to thank Xiao-Ming Lu for his comments on the revision
of present manuscript. This research is supported by National
Natural Science Foundation of China under Grant Nos. 11971140,
11701259, 11675113, Beijing Municipal Commission of Education
(KZ201810028042), and Beijing Natural Science Foundation (Z190005).
LZ is also supported by Zhejiang Provincial
Natural Science Foundation of China under Grant no. LY17A010027.
\end{acknowledgments}

%=============================================================================%
\appendix
%=============================================================================%
\section{POOF OF THE EQUIVALENCE OF $f(\bsa,\bsb)>1$ AND $g(\bsa,\bsb)>2$}\label{sect:app-A}
%=============================================================================%

This result is mentioned in \cite{Yu2010} without proof. We provide below the detailed proof
for completeness. We first present the following proposition.

\begin{prop}
If $\bsa,\bsb\in\real^3$ with $\abs{\bsa}\leqslant1$ and
$\abs{\bsb}\leqslant1$, then
$\abs{\bsa}^2+\abs{\bsb}^2\leqslant1+\inner{\bsa}{\bsb}^2$ if and
only if $\abs{\bsa+\bsb}+\abs{\bsa-\bsb}\leqslant2$.
\end{prop}

\begin{proof}
If $\abs{\bsa}^2+\abs{\bsb}^2\leqslant1+\inner{\bsa}{\bsb}^2$ for
$\abs{\bsa}\leqslant1$ and $\abs{\bsb}\leqslant1$, then
$\abs{\bsa}^2+\abs{\bsb}^2\pm
2\inner{\bsa}{\bsb}\leqslant1\pm2\inner{\bsa}{\bsb}+\inner{\bsa}{\bsb}^2$.
That is, $\abs{\bsa\pm\bsb}\leqslant 1\pm\inner{\bsa}{\bsb}$. Then
$\abs{\bsa+\bsb}+\abs{\bsa-\bsb}\leqslant2$.

Now conversely, if $\abs{\bsa+\bsb}+\abs{\bsa-\bsb}\leqslant2$ for
$\abs{\bsa}\leqslant1$ and $\abs{\bsb}\leqslant1$, then
$\abs{\bsa+\bsb}^2\leqslant (2-\abs{\bsa-\bsb})^2$, i.e.,
\begin{eqnarray*}
&&\abs{\bsa}^2+\abs{\bsb}^2+2\inner{\bsa}{\bsb}\\
&&\leqslant 4-4\abs{\bsa-\bsb}+
\abs{\bsa}^2+\abs{\bsb}^2-2\inner{\bsa}{\bsb},
\end{eqnarray*}
which is equivalent to
$\abs{\bsa-\bsb}\leqslant1-\inner{\bsa}{\bsb}$. Then
$(\abs{\bsa-\bsb})^2\leqslant(1-\inner{\bsa}{\bsb})^2$ implies
that $\abs{\bsa}^2+\abs{\bsb}^2\leqslant1+\inner{\bsa}{\bsb}^2$.
\end{proof}

From the above proposition, we see that $f(\bsa,\bsb)>1$ if and only if
$g(\bsa,\bsb)>2$, namely,
$\Set{(\bsa,\bsb):f(\bsa,\bsb)>1}=\Set{(\bsa,\bsb):g(\bsa,\bsb)>2}$.

%=============================================================================%
\section{PROBABILITY DENSITY FUNCTION OF THE INNER PRODUCT OF TWO INDEPENDENT RANDOM UNIT VECTORS}\label{sect:app-B}
%=============================================================================%

Recall that there is a unique unitary-invariant measure (up to
normalization) $\mu$ over the sphere:
\begin{eqnarray*}
\dif\mu(\bsu) = \frac1{N_m}\delta(1-\abs{\bsu})[\dif\bsu],
\end{eqnarray*}
where
\begin{eqnarray*}
N_m=\int_{\real^m}\delta(1-\abs{\bsu})[\dif\bsu] =
\frac{2\pi^{\frac m2}}{\Gamma(\frac m2)}.
\end{eqnarray*}
Now the probability density function of the inner product
$\inner{\bsu}{\bsv}$ between two independent random unit vectors
$\bsu$ and $\bsv$ can be expressed as
\begin{eqnarray*}
p_m(s) = \int\delta(s-\inner{\bsu}{\bsv})\dif\mu(\bsu)\dif\mu(\bsv).
\end{eqnarray*}
By using the Haar measure (also denoted $\mu$) over the orthogonal
group, the above integral can be rewritten as
\begin{eqnarray*}
p_m(s) &=& \int\delta(s-\inner{U\bse_1}{V\bse_1})\dif\mu(U)\dif\mu(V)\\
&=& \int\delta(s-\inner{\bse_1}{W\bse_1})\dif\mu(W)\\
&=& \int\delta(s-\inner{\bse_1}{\bsu})\dif\mu(\bsu),
\end{eqnarray*}
where $U$, $V$, and $W$ are some unitary operators.

Using the Lebesgue measure, we obtain that
\begin{eqnarray*}
p_m(s) &=& \frac1{N_m}\int\delta(s-u_1)\delta(1-\abs{\bsu})[\dif\bsu]\\
 &=&\frac2{N_m}\int_{\real^{m-1}}\delta\Pa{(1-s^2)-\sum^m_{j=2}u^2_j}\prod^m_{j=2}\dif u_j.
\end{eqnarray*}
Set
\begin{eqnarray*}
\psi(t)=\int_{\real^{m-1}}\delta\Pa{t-\sum^m_{j=2}u^2_j}\prod^m_{j=2}\dif
u_j.
\end{eqnarray*}
Then its Laplace transform is given by
\begin{eqnarray*}
L(\psi)(\omega):=\prod^m_{j=2}\int_\real e^{-\omega u^2_j}\dif
u_j=\Pa{\frac{\pi}{\omega}}^{\frac{m-1}2}.
\end{eqnarray*}
Hence
\begin{eqnarray*}
\psi(t) = L^{-1}\Pa{\Pa{\frac{\pi}{\omega}}^{\frac{m-1}2}}(t)
=\frac{\pi^{\frac{m-1}2}}{\Gamma\Pa{\frac{m-1}2}}t^{\frac{m-3}2}.
\end{eqnarray*}
Therefore we have

\begin{prop}
The probability density function of the inner product between $\bsu$
and $\bsv$ is given by $p_m(s)=\frac2{N_m}\psi(1-s^2)$, that is,
\begin{eqnarray*}
p_m(s) = C_m\cdot(1-s^2)^{\frac{m-3}2},\quad s\in[-1,1],
\end{eqnarray*}
where $C_m=\frac{\Gamma\Pa{\frac m2}}{\sqrt{\pi}\Gamma\Pa{\frac{m-1}2}}$.
\end{prop}

In particular, when $m=3$, $p_3(s)=\frac12$ for $s\in[-1,1]$.

Finally, the integral mentioned in the text is formulated as
\begin{eqnarray*}
&&\iint\Phi(\inner{\bsu}{\bsv})\delta(1-\abs{\bsu})\delta(1-\abs{\bsv})[\dif\bsu][\dif\bsv]\notag\\
&&=N^2_m\int^1_{-1}\Phi(s)p_m(s)\dif s,
\end{eqnarray*}
for any suitable function $\Phi(\cdot)$ of the inner product
$\inner{\bsu}{\bsv}$.

%=============================================================================%
\section{CALCULATION OF THE VOLUME: $\vol(\Theta^{\rN\rJ\rM}_{2,2})$}\label{sect:app-C}
%=============================================================================%

In fact,
\begin{eqnarray*}
\vol(\Theta^{\rN\rJ\rM}_{2,2}) = \int_{\Theta^{\rN\rJ\rM}_{2,2}}\dif
a_0\dif b_0[\dif\bsa][\dif\bsb],
\end{eqnarray*}
where a generic element $(a_0,a\cdot\bsu,b_0,b\cdot\bsv)$ of the set
$\Theta^{\rN\rJ\rM}_{2,2}\subset\Theta_{2,2}$ should satisfy
inequality \eqref{eq:yu}. Thus $\Theta^{\rN\rJ\rM}_{2,2}$ can be
transformed into the following form:
\begin{widetext}
\begin{eqnarray*}
\widetilde\Theta^{\rN\rJ\rM}_{2,2}=\Set{(a_0,a,b_0,b,s):
\begin{cases} \abs{a_0}+a\leqslant1,\abs{b_0}+b\leqslant1,\,\text{ where }\,a_0,b_0\in[-1,1],\,a,b\in[0,1]\\[1mm]
\Pa{1-h(a_0,a)^2-h(b_0,b)^2}\Pa{1-\frac{a^2_0}{h(a_0,a)^2}-\frac{b^2_0}{h(b_0,b)^2}}>
(abs-a_0b_0)^2\\[1mm]
h(x_0,x)=\frac{\sqrt{(1+x_0)^2-x^2}+\sqrt{(1-x_0)^2-x^2}}2\,\text{ and }\,s\in[-1,1]
\end{cases}}.
\end{eqnarray*}
\end{widetext}
Based on this observation, we get
\begin{eqnarray*}
\vol(\Theta^{\rN\rJ\rM}_{2,2})&=&
N^2_3\int_{\widetilde\Theta^{\rN\rJ\rM}_{2,2}}a^2b^2p_3(s)\dif
a_0\dif b_0\dif a\dif b \dif s\\
&=& 8\pi^2\int_{\widetilde\Theta^{\rN\rJ\rM}_{2,2}}a^2b^2\dif
a_0\dif b_0\dif a\dif b \dif s.
\end{eqnarray*}
By numerical calculation we have $\vol(\Theta^{\rN\rJ\rM}_{2,2}) \doteq 1.09662$.
%-----------------------------------------------------%
%=============================================================================%

\end{document}